\documentclass[reqno]{amsart}

\usepackage{graphicx}
\usepackage{amsfonts}
\usepackage{dsfont} 

\newtheorem{lemma}{Lemma}
\newtheorem{theorem}{Theorem}
\newtheorem{corollary}{Corollary}

\newtheorem{remark}{Remark}

\newcommand{\be}{\begin{eqnarray}}
\newcommand{\ee}{\end{eqnarray}}
\newcommand{\bee}{\begin{eqnarray*}}
\newcommand{\eee}{\end{eqnarray*}}
\newcommand{\R}{{\mathbb R}}

\newcommand{\Z}{{\mathbb Z}}

\newcommand{\W}{{W}}
\newcommand{\Hh}{{H}}
\newcommand{\I}{\mathds 1} 

\newcommand{\asy}{{\mathcal O}}

\newcommand{\da}{d_A}
\newcommand{\Vector}[1]{{\mathbf {#1}}}
\newcommand{\hhbar}{h}
\newcommand{\Chi}{\xi}
\newcommand{\w}{u}

\newcommand{\q}{q}

\begin{document}

\title [Tight-binding approximation for NLS] {Derivation of the tight-binding approximation for time-dependent nonlinear Schr\"odinger equations}

\author {Andrea SACCHETTI}

\address {Department of Physics, Informatics and Mathematics, University of Modena e Reggio Emilia, Modena, Italy.}

\email {andrea.sacchetti@unimore.it}

\date {\today}

\thanks {This paper is partially supported by GNFM-INdAM}

\begin {abstract} 
In this paper we consider the nonlinear one-dimensional time-dependent Schr\"odinger equation with a periodic potential and a bounded 
perturbation. \ In the limit of large periodic potential the time behavior of the wavefunction can be approximated, with a precise estimate 
of the remainder term, by means of the solution to the discrete nonlinear Schr\"odinger equation of the tight-binding model.

\bigskip

{\it Ams classification (MSC 2010):} 35Q55, 81Qxx, 81T25. 


\end{abstract}

\maketitle

\section {Introduction} \label {Sec0}

Here we consider the nonlinear one-dimensional time-dependent Schr\"odinger equation with a cubic nonlinearity, a  
periodic potential $V$ and a perturbing potential $W$ 
\be
\left \{
\begin {array}{l}
i \hbar \frac {\partial \psi}{\partial t} = - \frac {\hbar^2}{2m} \frac {\partial^2 \psi}{\partial x^2}  + 
\frac {1}{\epsilon} V \psi + \alpha_1 \W \psi + \alpha_2 |\psi |^{2 } \psi \, ,  \psi (\cdot , t ) \in L^2 (\R )  \\ 
\psi (x,0) = \psi_0 (x) 
\end {array}
\right.  \label {Eq1}
\ee
in the limit of large periodic potential, i.e. $0< \epsilon \ll 1$; $\alpha_1$ represents the 
strength of the perturbing potential $W$ and $\alpha_2$ represents the strength of the nonlinearity term. \ Equation (\ref {Eq1}) is the so called 
Gross-Pitaevskii equation for Bose-Einstein condensates where $\hbar$ is the Planck's constant and $m$ is the mass of the single atom. \ Such a model describes, for instance, 
one-dimensional Bose-Einstein condensates in an optical lattice and under the effect of an external field with potential $\alpha_1 W$; in particular, when such a perturbing 
potential is a Stark-type potential, that is it is locally linear, then recently has been shown the existence of Bloch oscillations for the wavefunction condensate and a precise measurement of 
the gravity acceleration has given \cite {FPST,PWTAPT}.

In the physical literature a standard way to study equation (\ref {Eq1}) consists in reducing it to a discrete Schr\"odinger equation 
taking into account only nearest neighbor interactions, the so called \emph {tight-binding model} \cite {AKKS}. \ The validity of such an approximation is, as 
far as we know, not yet rigorously proved in a general setting.

Recently, it has been proved that (\ref {Eq1}) admits a family of stationary solutions by reducing it to discrete nonlinear Schr\"odinger equations 
\cite {FS2,PSM,Sacchetti}. \ Concerning the reduction of the time-dependent equation to a discrete time-dependent nonlinear Schr\"odinger equation 
much less is known and rigorous results are only given under some conditions: for instance, in \cite {BS} the authors prove the validity of the reduction to discrete 
nonlinear Schr\"odinger equations for large times when $V$ is multiple-well trapped potential; while, 
in \cite {PS} a similar result for periodic potentials $V$ satisfying a sequence of specific technical 
conditions (see Theorem 2.5 \cite {P} for a resume) is obtained. \ We must also recall the papers \cite {P1,P2,P3} where applications of the orbital functions in a 
similar context is developed; in particular, in \cite {P2} the authors prove the validity of the reduction to discrete nonlinear Schr\"odinger equations 
of the Gross-Pitaevskii equation with a periodic linear potential and a sign-varying nonlinearity coefficient. \ In \cite {P3} the authors consider the case of a 
two-dimensional lattice; in particular, they show that tight-binding approximation is justified for simple and honeycomb lattices provided that the initial wavefunction is 
exponentially small.

In this paper we are able to show that the reduction of (\ref {Eq1}) to the time-dependent discrete nonlinear Schr\"odinger equations properly works 
with a precise estimate of the error, and that we don't need of special technical assumptions on the shape of the initial wavefunction and/or on the 
periodic potential; in fact, we have only to assume that the initial wavefunction is prepared on one band of 
the Bloch operator, let us say for argument's sake the first one. 

By introducing the new semiclassical parameter 
\bee
\hhbar = \hbar \sqrt {\epsilon /2m}  \, , 
\eee
the new time variable
\bee
\tau = \frac {\hbar}{h} t 
\eee
and the effective perturbation and nonlinearity strengths
\be
F = \alpha_1 \frac {2m\hhbar^2}{\hbar^2} \ \mbox { and } \  
\eta =\alpha_2 \frac {2m\hhbar^2}{\hbar^2} \, , \label {Eq1Bis}
\ee
then the above equation (\ref {Eq1}) takes the semiclassical form
\be
i \hhbar \frac {\partial \psi}{\partial \tau} = - {\hhbar^2} \frac {\partial^2 \psi}{\partial x^2}  +  V \psi + 
F \W \psi + \eta |\psi |^{2} \psi \label {Eq3}
\ee
with $\hhbar \ll 1$.

In the tight-binding approximation solutions to (\ref {Eq3}) are approximated by solutions to the time-dependent 
discrete nonlinear Schr\"odinger equation 
\be
i \hhbar \dot g_n =  - \beta (g_{n+1} + g_{n-1} ) + F \Chi_n g_n + \eta C_1 |g_n|^2 g_n \, ,\ n \in \Z \, , \label {TB}
\ee
where $\beta \sim e^{-S_0/\hhbar}$ is an \emph {exponentially small} positive constant in the semiclassical limit 
$\hhbar \ll 1$ (in fact, $S_0 >0$ is the Agmon distance between two adjacent wells, and for a precise estimate of the 
coupling parameter $\beta$ 
we refer to (\ref {Eq9Bis})). \ Furthermore, $\Chi_n = \langle \w_n , W \w_n \rangle $ and $C_1 = \| \w_n \|_{L^4}^4$ where, roughly speaking (a precise 
definition for $\w_n$ is given by \cite {Car,FS2,Sacchetti}), $\{ \w_n \}_{n \in \Z}$ is an orthonormal base of vectors of the eigenspace associated to the first 
band of the Bloch operator such that $\w_n \sim \psi_n$ as $\hhbar$ goes to zero; where $\psi_n$ is the ground state with associated energy $\Lambda_1$ of the 
Schr\"odinger equation with a single well potential $V_n$ obtained by filling all the wells, but the $n$-th one, of the periodic potential 
$V$:  $- \hhbar^2 \frac {\partial^2 \psi_n}{\partial x^2} + V_n \psi_n = \Lambda_1 \psi_n$. \ In fact, the linear operator 
$- \hhbar^2 \frac {\partial^2 }{\partial x^2} + V_n $ has a single well potential and thus it has a not empty discrete spectum, we denote by $\Lambda_1$ the 
first eigenvalue (which is independent on the index $n$ by construction).

We must underline that usually the tight-binding approximation is constructed by making use of the Wannier's functions instead of the vectors $\w_n$ 
\cite {AKKS,P}. \ Indeed, the decomposition by means of the Wannier's functions turns out to be more natural and it works for any range of $\hhbar$; on the other 
hand, the use of a suitable base $\{ \w_n \}_{n\in \Z}$ in the semiclassical regime of $\hhbar \ll 1$ has the great advantage that the vectors $\w_n$ are 
explicitly constructed by means of the semiclassical approximation. \ In fact, Wannier's functions may be approximated by such vectors $\w_n$ as pointed out 
by \cite {O}. 

The analysis of the discrete nonlinear Schr\"odinger equations (\ref {TB}) depends on the relative value of the perturbative parameters $F$ and $\eta$ with 
respect to the coupling parameter $\beta$. \ In this paper we consider two situations.

In the first case, named {\it model 1} corresponding to {Hypothesis 3a)}, we assume that $\alpha_1$ and $\alpha_2$ are fixed and independent of $\epsilon$. \ In such a 
case we have that $\beta \ll |F|$ and $\beta \ll C_1 |\eta |$ and then the analysis of (\ref {TB}) is basically reduced to the analysis of a system on infinitely many 
decoupled equations. \ Indeed, the perturbative terms with strength $F$ and $\eta$ dominate the coupling term with strength $\beta$ between the adjacent wells. \ In 
fact, this model has some interesting features; for instance, when $W$ represents a Stark-type perturbation then the analysis of the stationary solutions exhibits 
the existence of a cascade of bifurcations \cite {SacchettiPRE,Sacchetti}. \ On other hand, due to the fact that the perturbation is \emph {large}, when compared with 
the coupling term, the validity of the tight-binding approximation is justified only for time intervals rather small. 

In the second case, named {\it model 2} corresponding to {Hypothesis 3b)}, we assume that both $\alpha_1$ and $\alpha_2$ go to zero when $\epsilon$ goes to 
zero. \ In particular, we assume that 
\bee
F \sim C_1 \eta \sim \beta \, .
\eee
That is the perturbative terms are of the same order of the coupling term. \ In such a case the validity of the 
tight-binding 
approximation holds true for times of the 
order of the inverse of the coupling parameter $\beta$, that is the time interval is exponentially large.

We must remark that one could consider, in principle, other limits for $\alpha_1$ and $\alpha_2$ when $\hhbar$ goes to 
zero and Theorem \ref {Teo4} is very general and it holds true under different assumptions concerning $\alpha_1$ and $\alpha_2$ provided that $F = \asy (\hhbar^2 )$ and $\eta = \asy 
(\hhbar^2 )$. \ In fact, Hyp. 3a) and Hyp. 3b) represents, in some sense, two opposite situations concerning the choice of the parameters.

In \S \ref {Sec2} we state the assumptions on equation (\ref {Eq3}) and we state our main results in Corollaries 1 
and 2, they follow from a more technical 
Theorem \ref {Teo4} we state and prove in \S \ref {Sezione4}. \ In \S \ref {Sec2Bis} we prove a priori estimate of the wavefunction $\psi$ and of its 
gradient $\nabla \psi$. \ In \S \ref {Sec3} we formally construct the discrete nonlinear Schr\"odinger equations; in this Section we make use of some ideas already 
developed by \cite {FS2,Sacchetti} and we refer to these papers as much as possible. \ We must underline that in \cite {FS2,Sacchetti} the estimate of the 
remainder terms is given in the norm $\ell^1$, while in the present paper estimates in the norm $\ell^2$ are necessary and thus most of the material of Section 3, and 
in particular Lemmata 2, 3, 4, 5 and 6, is original and it cannot be simply derived from the papers quoted above. \ In \S \ref {Sezione4} we finally prove the validity 
of the 
tight-binding approximation with a precise estimate of the error in Theorem \ref {Teo4}, the method used is based on an idea already applied by \cite {Sacchetti2005} for a double-well model 
and now applied to a periodic potential; in particular, in \S \ref {Sec4.1} we consider the case where $\alpha_1$ and $\alpha_2$ are fixed, i.e. \emph {model 1}, 
and in \S \ref {Sec4.2} we consider the case where $\alpha_1$ and $\alpha_2$ goes to zero as $\epsilon$ goes to zero in a suitable way,  i.e. \emph {model 2}.

\section {Description of the model and main results} \label {Sec2}

\subsection {Assumptions}

Here, we consider the nonlinear Schr\"odinger equation (\ref {Eq1}) where the following assumptions hold true.

{\noindent {\bf Hypothesis 1.} {\it $V(x)$ is a smooth, real-valued, periodic and non negative function with period $a$, i.e. 
\bee
V(x) = V(x+a) \, , \ \forall x \in \R \, , 
\eee
and with minimum point $x_0 \in \left [ - \frac 12 a ,+\frac 12 a \right )$ such that 
\bee
V(x) > V(x_0) \, , \ \forall x \in \left [ - \frac 12 a ,+\frac 12 a \right ) \setminus \{ x_0 \} \, . 
\eee
For argument's sake we assume that $V(x_0 )=0$ and $x_0=0$.}}

\begin {remark}
We could, in principle, adapt our treatment to a more general case where $V(x)$ has more than one absolute minimum point in the interval 
$\left [ - \frac 12 a , + \frac 12 a \right )$. 
\end {remark}

{\noindent {\bf Hypothesis 2.} {\it The perturbation $\W(x)$ is a smooth real-valued function. \ We assume that $W\in L^\infty (\R )$.}}

Concerning the parameters $m$, $\hbar$, $\alpha_1$, $\alpha_2$ and $\epsilon$ we make the following assumption

{\noindent {\bf Hypothesis 3.} {\it We assume the limit of large periodic potential, i.e. $\epsilon$ is a real and positive parameter small enough
\bee
\epsilon \ll 1 \, . 
\eee
Concerning the other parameters we assume that:

\begin {itemize}
\item [\bf a)] The parameters $m$, $\hbar$, $\alpha_1$, $\alpha_2$ are real-valued and independent of $\epsilon$;
\end {itemize}
or
\begin {itemize}
\item [\bf b)] The parameters $m$, $\hbar$ are real-valued and independent of $\epsilon$ while the parameters $\alpha_1$, $\alpha_2$ are real-valued and 
they go to zero as $\epsilon$ goes to zero, in particular we assume that
\bee
\lim_{\epsilon \to 0^+} \frac {F}{\beta} \in \R \setminus \{ 0 \} \ \mbox { and } \ \lim_{\epsilon \to 0^+} \frac {C_1 \eta}{\beta} \in \R \setminus \{ 0 \}\, , 
\eee
where $F$ and $\eta$ are defined in (\ref {Eq1Bis}), and where the parameters $\beta$ and $C_1$ depend on $\epsilon$ (by means of $\hhbar$) and they are 
defined by (\ref {Eq9Bis}) and (\ref {Eq10}).
\end {itemize}
For argument's sake we assume in both cases that $\alpha_1 \ge 0$; hence $F \ge 0$.
}

\begin {remark} \label {RemDue}
In both cases we have that $0\le  F \le C \hhbar^2$ and $|\eta |\le C \hhbar^2$ for some positive constant $C$. \ In the case b), in particular, $F$ and $\eta$ 
are exponentially small when $\hhbar$ goes to zero. 
\end {remark}

Let $H_B$ be the Bloch operator formally defined on $L^2 (\R , dx)$ as
\be
H_B:=  - {\hhbar^2}\frac {d^2}{dx^2} + V \, . \label {Eq5}
\ee 
It is well known that this operator admits self-adjoint extension on the domain $H^2 (\R )$, still denoted by $H_B$, and its spectrum is given by bands: 
\bee
\sigma (H_B) = \cup_{\ell =1}^{\infty} [E^b_\ell , E^t_\ell] \, , \ \mbox { where } \ E^t_\ell \le E^b_{\ell+1} < E^t_{\ell+1} \, . 
\eee
The intervals $(E^t_\ell, E^b_{\ell+1})$ are named gaps; a gap may be empty, that is $E^b_{\ell +1}=
E^t_\ell$, or not. \ It is well known that in the case of one-dimensional crystals all the gaps are empty if, and 
only if, the periodic potential is a constant function. \ Because we assume that the periodic potential is not a 
constant function then one gap, at least, is not empty (for a review of Bloch operator we refer to \cite {RS}). \ In particular, when $\hhbar$ is small enough then the 
following asymptotic behaviors \cite {WK1,WK2}
\be
\frac {1}{C} \hhbar \le E_1^b \le C \hhbar \ \mbox { and } \ \frac {1}{C} \hhbar \le E_2^b - E_1^t \le C \hhbar \label {Eq6}
\ee
hold true for some $C>1$; hence, the first gap between $E_1^t$ and $E_2^b$ is not empty in the semiclassical limit.

Let $\Pi$ the projection operator associated to the first band $[E_1^b, E_1^t ]$ of 
$H_B$  and let $\Pi_\perp =\I - \Pi$. \ Let 
\be
\psi = \psi_1 + \psi_\perp \ \mbox { where } \psi_1 = \Pi \psi \ \mbox { and } \ \psi_\perp = \Pi_\perp 
\psi \, . \label {Eq7}
\ee
We assume that 

{\noindent {\bf Hypothesis 4.} {\it $\Pi_\perp \psi_0 =0$, where $\psi_0 (x) = \psi (x,0)$; that is the wave function $\psi$ is initially prepared on the first band. \ Through 
the paper we assume, for argument's sake, that $\psi_0$ is normalized, i.e. $\| \psi_0 \|_{L^2} =1$.}}

\subsection {Main results}

Here, we state our main results; they are a consequence of a rather technical Theorem \ref {Teo4} we postpone to \S \ref {Sezione4}. \ 
Let $\Vector {g} \in C (\R , \ell^2 (\Z ))$ be the solution to the tight-binding model, that is the discrete nonlinear Schr\"odinger equation (\ref {TB}); 
let $\psi (\tau , x)\in C (\R , H^1 (\R ))$ be the solution to the nonlinear Schr\"odinger equation (\ref {Eq3}) with initial condition 
$\psi_0 (x)= \sum g_n (0) \w_n (x)$.

\begin {corollary} \label {Cor1Bis}
Under the assumption {Hypothesis 3a)} we have that for any $\tau \in [0, C \hhbar^{-1/2} ]$ then 
\bee 
\left \| \psi (\tau , \cdot ) - \sum_{n\in \Z} g_n (\tau) e^{i\Lambda_1 \tau /\hhbar } u_n (\cdot ) \right \|_{L^2} \le C \hhbar^{1/2}
 \, .  
\eee
\end {corollary}

\begin {corollary} \label {Cor2}
Under the assumption {Hypothesis 3b)} we have that for any $\tau \in [0, C \beta^{-1} \hhbar ]$, 
where $\beta^{-1}$ is exponentially large as $\hhbar$ goes to zero, then  
\bee 
\left \| \psi (\tau , \cdot ) - \sum_{n\in \Z} g_n (\tau) e^{i\Lambda_1 \tau /\hhbar } u_n (\cdot ) \right \|_{L^2} \le C e^{-\zeta / \hhbar}
\eee
for some $C>0$ and $\zeta >0$ independent of $\hhbar$.
\end {corollary}

\begin {remark}
In \cite {PS} the estimate of the error was given in the energy norm, and even in \cite {Sacchetti} we used the 
$H^1$-norm. \ If one wants to extend the result of Corollary 1 to the $H^1$-norm it is clear that one has to pay a price; 
indeed, in the proof of Theorem 4 the term $\| u_0 \|_{H^1} \sim \hhbar^{-1/2}$ would appear instead of the term 
$\| u_0 \|_{L^2} =1$ and therefore the estimate of the error becames meaningless. \ On the other hand, this argument is not 
critical in the case of the extension of Corollary 2 to the $H^1$-norm because the term $\| u_0 \|_{H^1}$ is controlled by means 
of the exponentially small term $e^{-\zeta /\hhbar}$. \ In fact, we expect that Corollary 2 still hold true with the $H^1$-norm even 
if we don't dwell here with the detailed proof.
\end {remark}

\subsection {Notation and some functional inequalities} Hereafter, we denote by \linebreak $\| \cdot \|_{L^p}$, 
$p\in [+1,+\infty ]$, the usual norm of the Banach space 
$L^p (\R , dx)$; we denote by $\| \cdot \|_{\ell^p}$, $p\in [+1,+\infty ]$, the usual norm of the Banach space $\ell^p (\Z )$. 

Hereafter, we omit the dependence on $\tau$ in the wavefunctions $\psi$ and in the vectors $\Vector {c}$ when this fact does not cause misunderstanding.

By $C$ we denote a generic positive constant independent of $\hhbar$ whose value may change from line to line.

If $f$ and $g$ are two given quantities depending on the semiclassical parameter $\hhbar$, then by $f \sim g$ we mean that 
\bee
\lim_{\hhbar \to 0^+} \frac {f}{g} \in \R \setminus \{ 0 \} \, . 
\eee

Furthermore, we recall some well known results for reader's convenience:

\begin {itemize}
  
\item [-] {\bf One-dimensional Gagliardo-Nirenberg inequality} by 
\S B.5 \cite {P}: 
\bee
\| f \|_{L^p} \le C \| \nabla f \|_{L^2}^\delta \| f \|_{L^2}^{1-\delta} \, ,\ \delta = \frac 12 - \frac 1p = \frac {p-2}{2p} \, , \ p \in [2,+\infty ] \, ,
\eee

\item [-] {\bf Gronwall's Lemma} by Theorem 1.3.1 \cite {PA}: let $u(\tau )$ be a non negative and continuous 
function such that 
\bee
u(\tau ) \le \alpha (\tau ) + \int_0^{\tau} \delta (\q ) u(\q ) d\q  \, , \ \forall \tau \ge 0 \, , 
\eee
where $\alpha (\tau )$ and $\delta (\tau )$ are non negative and monotone not decreasing functions, then
\bee
u(\tau ) \le \alpha (\tau ) e^{\int_0^{\tau} \delta (\q ) d\q }\, , \forall \tau \ge 0 \, . 
\eee

\item [-] {\bf Agmon distance}: let $E$ be a given energy and $V(x)$ be a potential function, let $[z]_+ = z$ if $z\ge 0$ and $[z]_+ = 0$ if $z<0$; then the Agmon distance 
$\da (x,y)$ between two points $x,y \in \R^d$ is induced by the Agmon metric $[V(x) - E ]_+ dx^2$ where $dx^2$ is the standard metric on $L^2 (\R^d )$:
\bee
\da (x,y) = \inf_{\gamma \in {\mathcal C}} \int_0^1 \sqrt {\left [ V(\gamma (t))-E \right ]_+} |\gamma' (t)| dt 
\eee
where ${\mathcal C} $ is the set of piecewise paths $\gamma$ in $\R^d$ connecting $\gamma (0)=x$ and $\gamma (1)=y$ (see \cite {H} for a resume). \ In particular, in dimension $d=1$ 
and for energy $E=V_{min}$ we denote by $S_0=\int_{x_n}^{x_{n+1}} \sqrt {V(x)-V_{min}} dx$ the Agmon distance between the bottoms $x_n$ and $x_{n+1}$ of two adjacent wells; 
by the periodicity of $V(x)$ then $S_0$ does not depend on the index $n$.

\end {itemize}

\section {Preliminary results} \label {Sec2Bis}

We recall here some results by \cite {C1,C2,C3} concerning the solution to the time-dependent nonlinear Schr\"odinger equation 
with initial wavefunction $\psi_0$. \ The linear operator $\Hh $, formally defined as 
\bee
\Hh := H_B + F \W 
\eee
on the Hilbert space $L^2 (\R , dx)$, admits a self-adjoint extension on the domain $H^2 (\R )$, still denoted by $\Hh$. \ In order to discuss the local and global 
existence of solutions to (\ref {Eq3}) we apply Theorem 4.2 by \cite {C3}: if $\psi_0 \in H^1 (\R )$ there is a unique 
solution $\psi \in C([-T,T], H^1 (\R ))$ to (\ref {Eq3}) with initial datum $\psi_0$, such that 
\bee
\psi ,  \psi \frac {\partial (V+FW)}{\partial x}, \frac {\partial \psi}{\partial x} \in L^{8} ([-T,T];L^{4} (\R )) \, , 
\eee
for some $T>0$ depending on $\| \psi_0 \|_{H^1}$. 

In fact (see \cite {C2}), this solution is global in time for any $\eta \in \R$ (because in the case of one-dimensional nonlinear 
Schr\"odinger equations the cubic nonlinearity in sub-critical) and (\ref {Eq3}) 
enjoys the conservation of the mass
\bee
\| \psi (\cdot , \tau )\|_{L^2 } = \| \psi_0 (\cdot )\|_{L^2 }
\eee
and of the energy
\bee
{\mathcal E} \left [ \psi (\cdot , \tau ) \right ] = {\mathcal E} \left [ \psi_0 (\cdot )\right ]
\eee
where
\bee
{\mathcal E} (\psi ) &:=& \langle \Hh \psi , \psi \rangle + \frac {\eta}{2} \| \psi 
\|_{L^{4}}^{4} \\ 
&=& {\hhbar^2} \left \| \frac {\partial \psi }{\partial x}  \right \|_{L^2 }^2 + \langle V \psi , \psi \rangle + 
F \langle \W \psi , \psi \rangle + \frac {\eta}{2} \| \psi \|_{L^{4}}^{4}
\eee

Here, we prove some useful preliminary a priori estimates.

\begin {theorem} \label {Teo1}
The following a priori estimates hold true for any $\tau \in \R$:
\bee
\| \psi \|_{L^2} = \| \psi_0 \|_{L^2} =1 \ \mbox { and } \ \| \nabla \psi \|_{L^2} \le C \hhbar^{-1/2} \, , \\
\| \psi_1 \|_{L^2} \le \| \psi_0 \|_{L^2} =1 \ \mbox { and } \ \| \nabla \psi_1 \|_{L^2} \le C \hhbar^{-1/2}\, ,  \\
\| \psi_\perp \|_{L^2} \le \| \psi_0 \|_{L^2} =1 \ \mbox { and } \ \| \nabla \psi_\perp \|_{L^2} \le C \hhbar^{-1/2} \, ;
\eee
for some positive constant $C$.
\end {theorem}

\begin {proof}
From the conservation of the norm we have that
\bee
\| \psi_0 \|_{L^2}^2 = \| \psi \|_{L^2}^2 = \| \psi_\perp \|_{L^2}^2 + \| \psi_1 \|_{L^2}^2 \, ; 
\eee
hence 
\bee
\| \psi_\perp \|_{L^2} \le \| \psi_0 \|_{L^2} = 1 \ \mbox { and } \ \| \psi_1 \|_{L^2} \le \| \psi_0 \|_{L^2} = 1 \, . 
\eee

From the conservation of the energy we may obtain a priori estimate of the gradient of the wavefunction. \ Let
\bee
{\mathcal E}({\psi_0}) = \langle H_B\psi_0 , \psi_0 \rangle + F \langle W \psi_0 , \psi_0 \rangle + \frac 12 \eta \| \psi_0 \|_{L^4}^4 \, , 
\eee
where $\langle H_B\psi_0 , \psi_0 \rangle \sim \hhbar $ since $\psi_0$ is restricted to the eigenspace associated to the first band. \ Recalling that $V \ge 0$ 
then we have that
\bee
\hhbar^2 \| \nabla \psi_0 \|_{L^2}^2 \le \langle H_B \psi_0 , \psi_0 \rangle \sim \hhbar \, , 
\eee
which implies $\| \nabla \psi_0 \|_{L^2} \le C \hhbar^{-1/2}$. \ From this fact, using the fact that $W$ is a bounded potential and by the 
Gagliardo-Nirenberg inequality we have that
\bee
\| \psi_0 \|_{L^4}^4 \le C \| \nabla \psi_0 \|_{L^2} \| \psi_0 \|_{L^2}^3 \le C \| \nabla \psi_0 \|_{L^2} \le C \hhbar^{-1/2} \, . 
\eee
Hence, ${\mathcal E} (\psi_0 ) \sim \hhbar$ since $F \le C \hhbar^2$ and $|\eta | \le C \hhbar^2$ (see Remark \ref {RemDue}). \ Thus, the conservation of the energy 
implies the following inequality:
\bee
\hhbar^2 \| \nabla \psi \|_{L^2}^2 
&=& {\mathcal E}({\psi_0}) - \langle V \psi , \psi \rangle - F \langle W \psi , \psi \rangle - \frac 12 \eta \| \psi \|_{L^4}^4 \\ 
&\le & {\mathcal E}({\psi_0}) - V_{min} \| \psi \|_{L^2}^2 - F W_{min} \| \psi \|_{L^2}^2 - \frac 12 \eta \| \psi \|_{L^4}^4 \\ 
&\le & {\mathcal E}({\psi_0}) - F W_{min} - \frac 12 \eta \| \psi \|_{L^4}^4 
\eee
since $V_{min}=0$ and by the conservation of the norm. \ Let us set
\bee
\Lambda = \frac {{\mathcal E}({\psi_0}) - F W_{min} }{\hhbar^2} \ \mbox { and } \ \Gamma = \frac 12 \frac {\eta}{\hhbar^2} = \frac {m\alpha_2}{\hbar^2} \, , 
\eee
then $|\Gamma |\le C$ and $\Lambda \sim \hhbar^{-1}$ as $\hhbar$ goes to zero. \ The previous inequality becomes 
\bee
\| \nabla \psi \|_{L^2}^2 \le |\Lambda | + |\Gamma | \, \| \psi \|_{L^4}^4 \, . 
\eee
Again, the Gagliardo-Nirenberg inequality implies that 
\bee
\| \psi \|_{L^4}^4 \le C \| \nabla \psi \|_{L^2} \| \psi \|_{L^2}^3 = C \| \nabla \psi \|_{L^2}  
\eee
and thus we get
\bee
\| \nabla \psi \|_{L^2}^2 \le |\Lambda | + |\Gamma | C  \| \nabla \psi \|_{L^2}  
\eee
from which it follows that
\bee
\| \nabla \psi \|_{L^2} \le \frac 12 \left [ |\Gamma |C + \sqrt {\Gamma^2 C^2 + 4 |\Lambda |} \right ] \le C |\Lambda |^{1/2} \le C \hhbar^{-1/2}
\eee
for some positive constant $C$. 

Since $\Pi H_B = H_B \Pi$, we have that 
\bee
&& {\mathcal E} (\psi_0) - F \langle W\psi , \psi \rangle - \frac 12 \eta \| \psi \|_{L^4}^4 = \langle H_B \psi , \psi \rangle = 
\langle H_B \psi_1 , \psi_1 \rangle + \langle H_B \psi_\perp , \psi_\perp \rangle  \\ 
&& \ = \hhbar^2 \| \nabla \psi_1 \|_{L^2}^2 + \hhbar^2 \| \nabla \psi_\perp \|_{L^2}^2 
+ \langle V \psi_1 , \psi_1 \rangle + \langle V \psi_\perp , \psi_\perp \rangle \ge \hhbar^2 \| \nabla \psi_1 \|_{L^2}^2
\eee
since $V_{min}\ge 0$. \ Then,
\bee
\hhbar^2 \| \nabla \psi_1 \|_{L^2}^2 \le  
C \hhbar + \frac 12 | \eta |\, \| \psi \|_{L^4}^4 
\le  C \hhbar + C |\eta | \| \nabla \psi \|_{L^2} \le C \hhbar \, ; 
\eee
hence, 
\bee
\| \nabla \psi_1 \|_{L^2} \le C \hhbar^{-1/2} \, . 
\eee
Similarly we get
\bee
\| \nabla \psi_{\perp} \|_{L^2} \le C \hhbar^{-1/2} \, , 
\eee
and thus the proof of the Theorem is so completed.
\end {proof}

\begin {corollary} \label {Cor1}
We have the following estimates:
\bee
\| \psi \|_{L^\infty} \le C \hhbar^{-1/4}\, , \ \| \psi_1 \|_{L^\infty} \le C \hhbar^{-1/4}\, , \ 
\| \psi_\perp \|_{L^\infty} \le C \hhbar^{-1/4}\, , \ \forall \tau \ge 0 \, .
\eee
\end {corollary}

\begin {proof}
They immediately follow from the one-dimensional Gagliardo-Nirenberg inequality (where $p=+\infty$ and $\delta = \frac 12$) and from the previous result.
\end {proof}

\section {Construction of the discrete time-dependent nonlinear Schr\"odinger equation} \label {Sec3}

By the Carlsson's construction \cite {Car} resumed and expanded by \cite {FS2} (see also \S 3 \cite {Sacchetti} for a short review of the main results) 
we may write $\psi_1$ by means of a linear combination of a suitable orthonormal base $\{ u_n \}_{n \in \Z}$ of the space $\Pi \left [ L^2 (\R ) \right ]$, that is
\be
\psi_1 (x) = \sum_{n\in \Z} c_n \w_n (x) \, , \label {Eq8}
\ee
where  $u_n \in H^1 (\R )$ and $ \Vector {c} = \{ c_n \}_{n \in \Z} \in \ell^2 (\Z ) $ and where we omit, for simplicity's sake, the dependence on $\tau$ in the 
wavefunctions $\psi$, $\psi_1$, $\psi_\perp$ as well as in the vector $\Vector {c}$.

By inserting (\ref {Eq7}) and (\ref {Eq8}) in equation (\ref {Eq3}) then it takes the form (where $\dot {} = \frac {\partial }{\partial \tau}$) 
\be
\left \{ 
\begin  {array}{lcl}
i \hhbar \dot c_n &=& \langle \w_n , H_B \psi \rangle + F \langle \w_n , \W \psi \rangle + \eta \langle \w_n , 
|\psi |^{2} \psi  \rangle \, , \ n \in \Z \\ 
i\hhbar \dot \psi_\perp &=& \Pi_\perp H_B \psi + F \Pi_\perp \W \psi + \eta \Pi_\perp |\psi |^{2} \psi  
\end  {array}
\right.  \, , \label {Eq8Bis}
\ee
where $\Vector {c} \in \ell^2$ and $\psi_\perp$ are such that for any $\tau \in \R$
\bee
\| \psi_\perp \|_{L^2} \le \| \psi_0 \|_{L^2} =1 \ \mbox { and } \ \sum_{n\in \Z} |c_n|^2 = \| \Vector {c} \|^2_{\ell^2} = \| \psi_1 \|^2_{L^2} 
\le \| \psi_0 \|^2_{L^2} =1 \, . 
\eee

By mean of the gauge choice $\psi (x,\tau ) \to e^{i \Lambda_1 \tau /\hhbar} \psi (x,\tau )$, and then 
$\psi_\perp (x,\tau ) \to e^{i \Lambda_1 \tau /\hhbar} \psi_\perp (x,\tau )$ 
and $c_n (\tau ) \to e^{i \Lambda_1 \tau /\hhbar} c_n (\tau )$, (\ref {Eq8Bis}) takes the form
\be
\left \{ 
\begin  {array}{lcl}
i \hhbar \dot c_n &=& \langle \w_n , H_B \psi \rangle - \Lambda_1 c_n + F \langle \w_n , \W \psi \rangle + \eta \langle \w_n , 
|\psi |^{2} \psi  \rangle \, , \ n \in \Z \\ 
i\hhbar \dot \psi_\perp &=& \Pi_\perp (H_B-\Lambda_1) \psi + F \Pi_\perp \W \psi + \eta \Pi_\perp |\psi |^{2} \psi  
\end  {array}
\right.  \, , \label {Eq9}
\ee
where $\Lambda_1$ is the energy associated to the ground state of the Schr\"odinger operator $- \hhbar^2 \frac {\partial^2}{\partial x^2} + V_n$, with single 
well potential $V_n$ obtained by filling all the wells of the periodic potential $V$, but the $n$-th one; since $V_n (x) = V_m (x-x_n+x_m)$ by construction 
(see \cite {FS2,Sacchetti} for details) then the spectrum of this linear operator is independent on the index $n$ and the eigenvetor $\psi_n$ associated to the 
ground state $\Lambda_1$ is such that $\psi_m (x) = \psi_n (x-x_m+x_n)$ .

We have that
\bee
\langle \w_n , H_B \psi \rangle = \Lambda_1 c_n - \beta (c_{n+1} + c_{n-1}) + r_{1,n} \, , 
\eee
where $\Lambda_1$ and $\beta$ are independent of the index $n$ and $\beta $ is such that for any $0<\rho <S_0$ there is $C:=C_\rho$ such that
\be
\frac {1}{C} e^{-(S_0+ \rho )/\hhbar} < \beta < C e^{-(S_0- \rho )/\hhbar} \, ; \label {Eq9Bis}
\ee
the remainder term $r_{1,n}$ is defined as  
\bee
\ r_{1,n} := \sum_{m\in \Z} \tilde D_{n,m} c_m 
\eee
where $\tilde D_{n,m}$ satisfies Lemma 1 in \cite {Sacchetti}. \ Furthermore, 
\bee
\langle \w_n , \W \psi \rangle =  \Chi_n c_n + r_{2,n} + r_{3,n} \, ,  
\eee
where we set 
\bee
\Chi_n = \langle u_n , W u_n \rangle \, , \ r_{2,n} = \sum_{m\in \Z \ : \ m \not= n} \langle u_n , W u_m \rangle c_m \mbox { and }  
r_{3,n} = \langle u_n, W\psi_\perp \rangle \, . 
\eee
Finally
\bee
\langle \w_n , |\psi|^{2} \psi \rangle = C_1 |c_n|^{2} c_n + r_{4,n} \, , \ C_1 = \| u_n \|_{L^4}^4 \, , 
\eee
where we set
\bee
r_{4,n} = \langle \w_n , |\psi|^{2} \psi \rangle - C_1 |c_n|^{2} c_n 
\eee
and where by Lemma 1.vi \cite {Sacchetti} it follows that 
\be
C_1 = \| \w_n \|_{L^4}^{4} \equiv  \| \w_0 \|_{L^4}^{4} \sim \hhbar^{- {1}/{2}} \mbox { as } \hhbar \mbox { goes to zero.} \label {Eq10}
\ee

Therefore, (\ref {Eq9}) may be written 
\be
\left \{ 
\begin  {array}{lcl}
i \hhbar \dot c_n &=&  - \beta (c_{n+1} + c_{n-1}) + F \Chi_n c_n + \eta C_1 |c_n|^2 c_n + r_n \\ 
i\hhbar \dot \psi_\perp &=& \Pi_\perp (H_B-\Lambda_1 ) \psi + F \Pi_\perp \W \psi + \eta \Pi_\perp |\psi |^{2} \psi  
\end  {array}
\right.  \, , \label {Eq11}
\ee
where we set
\bee
r_n = r_{1,n} + F r_{2,n} + F r_{3,n} + \eta r_{4,n} \, . 
\eee
Tight-binding approximation (\ref {TB}) is obtained by putting $\psi_\perp \equiv 0 $ and by neglecting the coupling term $r_n$ in (\ref {Eq11}).

We have the following estimates.

\begin {lemma}
\bee
\| \Vector {r}_1 \|_{\ell^2 } \le C e^{-(S_0 + \zeta )/\hhbar } \| \Vector {c} \|_{\ell^2 }
\eee
for some positive constants $C$ and $\zeta$ independent of $\hhbar$.

\end {lemma}

\begin {proof}
Such an estimate directly comes from Lemma 1 by \cite {Sacchetti}.
\end {proof}

\begin {lemma} \label {Lem2}
For any $0 < \rho <S_0$ there is a positive constant $C:=C_\rho$ such that 
\bee 
\| \Vector {r}_2 \|_{\ell^2 } \le C e^{(S_0 - \rho )/\hhbar} \| \Vector {c} \|_{\ell^2 }\, . 
\eee 
\end {lemma}

\begin {proof}
We set
\bee
W_{n,m} = 
\left \{ 
\begin {array}{ll}
\langle u_n , W u_m \rangle & \ \mbox { if } \ n \not= m \\ 
0 & \ \mbox { if } \ n = m
\end {array}
\right. \, ; 
\eee
then $r_{2,n} = \sum_{m\in \Z} W_{n,m} c_m $. \ By Example 2.3 \S III.2 \cite {Kato} it follows that 
\bee 
\| \Vector {r}_2 \|_{\ell^2 } \le \max [M',M''] \| \Vector {c} \|_{\ell^2 }
\eee 
where $M'$ and $M''$ are such that $\sum_{m\in \Z} |W_{n,m} | \le M'$ and $\sum_{n\in \Z} |W_{n,m} | \le M''$ for any $n \in \Z$; then $M' = M''$ because 
$|W_{n,m} | = |W_{m,n} |$. \ Since $W$ is a bounded operator and by Lemma 1.iv \cite {Sacchetti} it immediately follows that $M' = C e^{(S_0 - \rho )/\hhbar}$ 
for any $0 < \rho < S_0 $ and for some positive constant $C:=C_\rho$. \ Hence, Lemma \ref {Lem2} is so proved. \end {proof}

\begin {lemma} \label {Lem3} 
\bee
\| \Vector {r}_3 \|_{\ell^2 } \le C  \| \psi_\perp \|_{L^2} \, . 
\eee
\end {lemma}

\begin {proof}
Since $r_{3,n} = \langle u_n, W \psi_\perp \rangle_{L^2}$ where $\{ u_n \}_{n \in \Z}$ is an orthonormal base of the space $\Pi \left [ L^2 (\R ) \right ]$; then, from 
the Parseval's identity it follows that 
\bee
\| \Vector {r}_3 \|_{\ell^2 } = \| \Pi W \Pi_\perp \psi \|_{L^2} = \| \Pi W \Pi_\perp \psi_\perp \|_{L^2} \le \| \Pi W \Pi_\perp \| \, \| \psi_\perp \|_{L^2} \le  
C  \| \psi_\perp \|_{L^2}
\eee
because $\Pi W \Pi_\perp$ is a bounded potential.
\end {proof}


For what concerns the vector $\Vector {r}_4$ let
\bee 
r_{4,n} =\langle \w_n , |\psi|^{2} \psi \rangle - C_1 |c_n|^{2} c_n = A_n + B_n
\eee
where we set 
\bee
A_n = \langle \w_n , |\psi|^{2} \psi \rangle - \langle \w_n , |\psi_1|^{2} \psi_1 \rangle 
\eee
and
\be
B_n = \langle \w_n , |\psi_1|^{2} \psi_1 \rangle - C_1 |c_n|^{2} c_n = \sum_{j,\ell ,m\in \Z}^\star \langle u_n , \bar u_m u_\ell u_j \rangle \bar c_m c_\ell c_j 
\label {Eq12}
\ee
where $\sum_{j,m,\ell \in \Z}^\star$ means that at least one of three indexes $j$, $\ell$ and $m$ is different from the index $n$.

\begin {lemma} \label {Lem4} Let $\Vector {B} = \{ B_n \}_{n\in \Z}$, then for any $0 < \rho < S_0$ there is a positive constant $C$ such that 
\bee
\| \Vector {B} \|_{\ell^2} \le C e^{-(S_0 - \rho )/\hhbar} \, . 
\eee
\end {lemma}

\begin {proof}
For argument's sake let us assume that $m$ is the index different from the index $n$ in the sum (\ref {Eq12}); then we have to check the term
\bee
 \sum_{m, \ell, j \in \Z ,\, m \not= n} \langle u_m u_n , u_\ell u_j \rangle \bar c_m c_\ell c_j = B_{1,n} + B_{2,n} + B_{3,n}
\eee
where
\bee
B_{1,n} &=& \sum_{j,m,\ell \in \Z}^{\star 1} \langle u_m u_n , u_\ell u_j \rangle \bar c_m c_\ell c_j  := \sum_{m \in \Z \setminus \{ n\}} 
\sum_{\ell \in \Z \setminus \{ m, n\}} \sum_{j \in \Z \setminus \{\ell, m, n\}} \langle u_m u_n , u_\ell u_j \rangle \bar c_m c_\ell c_j  \\
B_{2,n} &=& \sum_{j,m,\ell \in \Z}^{\star 2} \langle u_m u_n , u_\ell^2 \rangle \bar c_m c_\ell^2 := \sum_{m \in \Z \setminus \{ n\}} 
\sum_{\ell \in \Z \setminus \{ m, n\}} \langle u_m u_n , u_\ell^2 \rangle \bar c_m c_\ell^2 \\ 
B_{3,n} &=& \sum_{j,m,\ell \in \Z}^{\star 3} \langle u_m u_n , u_m^2 \rangle \bar c_m c_m^2  := 
\sum_{m \in \Z ,\,  m \not= n} \langle u_n u_m , u_m^2 \rangle \bar c_m c_m^2 
\eee

Let $0 < \rho < S_0$ be fixed; from Lemma 1.iv \cite {Sacchetti} it follows that for any $\rho' ,\rho'' >0$ such that $\rho' + \rho'' < \rho$ then 
there exists a positive constant $C>0$, independent of the indexes $n$ and $m$ and of the semiclassical parameter $\hhbar$, such that
\be
\left \| \w_n \w_m \right \|_{L^1 (\R )} \le C e^{ \left [ (S_0 -\rho ')|m-n| - \rho'' \right ] /\hhbar } \, . \label {Eq13Bis}
\ee
Now, observing that $|c_m| \le 1$ since $\| \Vector {c} \|_{\ell^2} \le 1$, then 
\bee
|B_{3,n}| &\le & \sum_{m \in \Z ,\,  m \not= n}| \langle u_n u_m , u_m^2 \rangle |\, | c_m|^3 \\  
&\le & \sum_{m \in \Z ,\,  m \not= n} \| u_n u_m \|_{L^1} \| u_m \|_{L^\infty}^2  \, | c_m|^2 \\
&\le & \sum_{m \in \Z ,\,  m \not= n} C \hhbar^{-1/2} e^{-[(S_0-\rho') |n-m|-\rho'']/\hhbar} \, | c_m|^2
\eee
where we make use of the estimate (\ref {Eq13Bis}) and where $\rho' , \rho '' >0$ are such that $\rho' + \rho'' < \rho $. \ Hence, 
\bee
\| \Vector {B}_3 \|_{\ell^2} &\le & \| \Vector {B}_3 \|_{\ell^1} \le  \sum_{n,m\in \Z \, , \ m\not= n} C \hhbar^{-1/2} e^{-[(S_0-\rho') |n-m|-\rho'']/\hhbar} | c_m|^2 \\ 
&\le & C e^{-(S_0-\rho) /\hhbar} \sum_{m\in \Z}  | c_m|^2 = C e^{-(S_0-\rho) /\hhbar} \| \Vector {c} \|_{\ell^2}^2 \\
&\le & C e^{-(S_0-\rho) /\hhbar} \, . 
\eee
for some positive constant $C$. \ For what concerns the term $B_{2,n}$ we have that
\bee
|B_{2,n}| &=& \left | \sum_{m, \ell \in \Z }^{\star 2} \langle u_m u_n , u_\ell^2 \rangle \bar c_m c_\ell^2 \right | \le 
\sum_{m, \ell \in \Z }^{\star 2} \langle |u_m | \, |u_n| , |u_\ell|^2 \rangle \, | c_m |\, |c_\ell |^2 \\ 
&\le &\sum_{ \ell \in \Z ,\, \ell \not= n}   \left \langle  |u_n |\sum_{m\in \Z} |u_m | , | u_\ell |^2 \right \rangle  |c_\ell |^2 \\ 
&\le & \max_{\ell \in \Z} \| u_\ell \|_{L^\infty} \left \| \sum_{m\in \Z} |u_m| \right \|_{L^\infty} \sum_{ \ell \in \Z ,\, \ell \not= n}   
\left \langle  |u_n | , | u_\ell | \right \rangle  |c_\ell |^2 \\
&\le & \max_{\ell \in \Z} \| u_\ell \|_{L^\infty} \left \| \sum_{m\in \Z} |u_m| \right \|_{L^\infty} \sum_{ \ell \in \Z ,\, \ell \not= n}   
C e^{-[(S_0-\rho' )|n-\ell |-\rho'']/\hhbar} |c_\ell |^2 \\ 
&\le & C \| u_0 \|_{L^\infty} \left \| \sum_{m\in \Z} |u_m| \right \|_{L^\infty} \sum_{ \ell \in \Z ,\, \ell \not= n}   
e^{-[(S_0-\rho' )|n-\ell |-\rho'']/\hhbar} |c_\ell |^2 \\ 
&\le & C \hhbar^{-3/4} \sum_{ \ell \in \Z ,\, \ell \not= n} e^{-[(S_0-\rho' )|n-\ell |-\rho'']/\hhbar} |c_\ell |^2
\eee 
since $\| u_\ell \|_{L^\infty} = \| u_0 \|_{L^\infty} \le C \hhbar^{-1/4}$ and $\left \| \sum_{m\in \Z} |u_m| 
\right \|_{L^\infty} \le C \hhbar^{-1/2}$ (see Lemma 1 \cite {Sacchetti}), from which it follows that
\bee
\| \Vector {B}_2 \|_{\ell^2} &\le & \| \Vector {B}_2 \|_{\ell^1} = \sum_{n\in \Z} |B_{2,n}| 
\le C \hhbar^{-3/4} \sum_{n,\ell \in \Z \, \ n \not= \ell}  
 Ce^{-[(S_0-\rho' )|n-\ell |-\rho'']/\hhbar} |c_\ell |^2 \\ 
 &\le & C e^{-(S_0-\rho )/\hhbar} \| \Vector {c} \|_{\ell^2} \le C e^{-(S_0-\rho )/\hhbar} \, ,
\eee 
where $\rho' , \rho '' >0$ are such that $\rho' + \rho'' < \rho < S_0$. \ Finally,
\bee
|B_{1,n}| &\le & \sum_{m, \ell, j \in \Z }^{\star 1} \langle |u_m|\, |u_n| , |u_\ell|\, |u_j| \rangle |c_m|\, |c_\ell|\, |c_j| \\ 
&\le & \frac 12 \sum_{m, \ell, j \in \Z }^{\star 1} \langle |u_m|\, |u_n| , |u_\ell|\, |u_j| \rangle |c_m|\, [ |c_\ell|^2 + |c_j|^2 ] \\ 
&\le & \sum_{m, \ell, j \in \Z }^{\star 1} \langle |u_m|\, |u_n| , |u_\ell|\, |u_j| \rangle |c_m| \, |c_j|^2  \\
&\le & \sum_{m \in \Z \setminus \{ n \},  j \in \Z \setminus \{ m,n \}} \left \langle |u_m|\, |u_n| , |u_j|\sum_{\ell \in \Z} |u_\ell| \right \rangle |c_m| \, |c_j|^2  \\ 
\eee
\bee
\ \ \ \ \ \ &\le & \left \| \sum_{\ell \in \Z} |u_\ell| \right \|_{L^\infty} \sum_{m \in \Z \setminus \{ n \},  j \in \Z \setminus \{ m,n \}} 
\langle |u_m|\, |u_n| , |u_j| \rangle |c_m| \, |c_j|^2  \\
&\le & \left \| \sum_{\ell \in \Z} |u_\ell| \right \|_{L^\infty} \sum_{  j \in \Z ,\, j \not= n} \left \langle |u_n| 
\sum_{m\in \Z} |u_m|  , |u_j| \right \rangle  |c_j|^2  \\
&\le & \left \| \sum_{\ell \in \Z} |u_\ell| \right \|_{L^\infty}^2 \sum_{  j \in \Z ,\, j \not= n}  \langle |u_n| , |u_j| \rangle  |c_j|^2  \\
&\le &  C \hhbar^{-1} \sum_{  j \in \Z ,\, j \not= n}  e^{-\left [ (S_0 - \rho') |n-j|- \rho'' \right ] /\hhbar } |c_j|^2  
\eee
since $|c_m |\le 1$. \ Hence, 
\bee
\| \Vector {B}_1 \|_{\ell^2} &\le & \| \Vector {B}_1 \|_{\ell^1} \le  C \hhbar^{-1} \sum_{n,j\in \Z \, , \ j\not= n} C 
e^{-\left [ (S_0 - \rho') |n-j|- \rho'' \right ] /\hhbar } |c_j|^2 \\
&\le & C e^{-(S_0 - \rho)  /\hhbar } \sum_{j\in \Z }|c_j|^2 =  C e^{-(S_0 - \rho)  /\hhbar } \| \Vector {c} \|_{\ell^2} \\
&\le & C e^{-(S_0 - \rho)  /\hhbar }\, .
\eee
From these estimates it follows that 
\bee
\| \Vector {B} \|_{\ell^2} \le C \left [ \| \Vector {B}_1 \|_{\ell^2} + \| \Vector {B}_2 \|_{\ell^2}+ \| \Vector {B}_3 \|_{\ell^2} \right ] \le C e^{-(S_0-\rho)/\hhbar }  
\eee
and Lemma \ref {Lem4} is so proved. 
\end {proof}

Now we deal with the vector $\Vector {A}$ with elements
\bee
A_n = \langle u_n , g \rangle 
\eee
where 
\bee
g := |\psi |^2 \psi - |\psi_1 |^2 \psi_1 = \bar \psi_1 \psi_\perp^2 + 2 |\psi_1|^2 \psi_\perp + \psi_1^2 \bar \psi_\perp + |\psi_\perp|^2 \psi_\perp + 
2 \psi_1 |\psi_\perp|^2 
\eee

\begin {lemma}\label {Lem5}
\bee
\| \Vector {A} \|_{\ell^2} \le C \hhbar^{-1/2} \| \psi_\perp \|_{L^2}
\eee
\end {lemma}

\begin {proof}
Indeed, since $\{ u_n \}_{n \in \Z}$ is an orthonormal base of $\Pi \left [ L^2 (\R )\right ]$ from the Parseval's identity it follows that
\bee
\| \Vector {A} \|_{\ell^2} &=& \| \Pi g \|_{L^2} \le \| g \|_{L^2} \le C \left [ \| \psi_\perp^3 \|_{L^2} + \| \psi_\perp^2 \psi_1 \|_{L^2} +  
\| \psi_\perp \psi_1^2 \|_{L^2} \right ] \\ 
&\le & C \max \left ( \| \psi_\perp \|_{L^\infty}^2 , \| \psi_1 \|_{L^\infty}^2 \right ) \| \psi_\perp \|_{L^2} \le C \hhbar^{-1/2} \| \psi_\perp \|_{L^2}
\eee
from Corollary \ref {Cor1}.
\end {proof}

In conclusion we have proved the following Lemma;

\begin {lemma} \label {Lem6}
\bee
\| \Vector {r}_4 \|_{\ell^2} \le  C e^{-(S_0 - \rho )/\hhbar} + C \hhbar^{-1/2} \| \psi_\perp \|_{L^2} \, . 
\eee
\end {lemma}

\section {Validity of the tight-binding approximation} \label {Sezione4}

First of all we need of the following estimate:

\begin {lemma} \label {Lem7}
Let us set 
\bee
\lambda := F \hhbar^{-1} + |\eta | \hhbar^{-3/2}  
\eee
and let $\Vector {c}$ and $\psi_\perp$ be the solutions to (\ref {Eq11}); then 
\be
\| \dot {\Vector {c}} \|_{\ell^2} \le \frac {C}{\hhbar} \max \left [ \beta ,\hhbar \lambda , \| \Vector {r} \|_{\ell^2} \right ]\, . \label {Eq15}
\ee
\end {lemma}

\begin {proof}
Indeed, from (\ref {Eq11}) it immediately follows that 
\bee
\| \dot {\Vector {c}} \|_{\ell^2}^2 & = & \frac {1}{\hhbar^2} \sum_{n\in \Z} \left | - \beta (c_{n+1} + c_{n-1} ) + F \Chi_n c_n + \eta C_1 |c_n|^2 c_n + r_n \right |^2  \\ 
&\le & \frac {C}{\hhbar^2} \left [ \beta^2 \sum_{n\in \Z} |c_{n+1}|^2  + \beta^2 \sum_{n\in \Z} |c_{n-1}|^2  + F^2 \max_{n\in \Z} |\Chi_n |^2  \sum_{n\in \Z} |c_{n}|^2 + \right. \\
&& \left. \ \ + \eta^2 C_1^2 \sum_{n\in \Z} |c_{n}|^6 + \sum_{n\in \Z} |r_n|^2 \right ] \\ 
&\le & \frac {C}{\hhbar^2} \left [ (\beta^2 + F^2 + \eta^2 C_1^2)  \| {\Vector {c}} \|_{\ell^2}^2 +  \| {\Vector {r}} \|_{\ell^2}^2 \right ]
\eee
from which the estimate (\ref {Eq15}) follows since $\| {\Vector {c}} \|_{\ell^2} \le 1 $ and $|c_n| \le 1$, $|\xi_n| \le C$ because $W$ is a bounded 
potential and $C_1 \sim \hhbar^{-1/2}$. 
\end {proof}

Hereafter, we denote by $\omega$ a quantity, whose value may change from line to line, such that 
\bee
|\omega | \le C e^{-(S_0 + \zeta )/\hhbar}  
\eee
for some $\zeta >0$ and some $C>0$ independent of $\hhbar$.

\begin {theorem} \label {Teo2}
\be
\| \Vector {r} \|_{\ell^2} \le \omega + C \lambda e^{- (S_0 - \rho )/\hhbar} + C\hhbar \lambda  \| \psi_\perp \|_{L^2} \label {F1}
\ee
\end {theorem}

\begin {proof}
Indeed, collecting the results from Lemmata \ref {Lem2},\ref {Lem3} and \ref {Lem6} and from Remark \ref {RemDue} we have that 
\bee
\| \Vector {r} \|_{\ell^2} &\le & \| \Vector {r_1} \|_{\ell^2} + F \, \| \Vector {r_2} \|_{\ell^2} + F\, \| \Vector {r_3} \|_{\ell^2} 
+|\eta |\, \| \Vector {r_4} \|_{\ell^2} \\ 
&\le & C e^{-(S_0 + \zeta)/\hhbar} + C F e^{-(S_0-\rho )/\hhbar} + C F  \| \psi_\perp \|_{L^2} + \\ 
&& \ \ + C |\eta | e^{-(S_0-\rho )/\hhbar} + C |\eta | \hhbar^{-1/2} \| \psi_\perp \|_{L^2} 
\eee
from which the statement immediately follows.
\end {proof}

Since $\psi_\perp (x,0) = \Pi_\perp \psi_0 =0$, then the second differential equation of the system (\ref {Eq11}) may be written as an integral equation of the 
Duhamel's form
\be
\psi_\perp (\tau ) = -i \int_0^\tau e^{-i(H_B-\Lambda_1 ) (\tau -\q )/\hhbar } \left [ \frac {F}{\hhbar} \Pi_\perp \W \psi 
+ \frac {\eta}{\hhbar} \Pi_\perp |\psi^{2}| \psi \right ] d\q  \label {Eq15Ter}
\ee

\begin {theorem} \label {Teo3}
We have the following estimate 
\be
\| \psi_\perp \|_{L^2} \le  \left \{ C \lambda + \tau C \hhbar^{-1} \lambda  \max \left [ \beta ,  \hhbar \lambda
\right ] \right \} e^{C \lambda  \tau}  \, , \ \forall \tau \in \R \, .  \label {F2}
\ee
\end {theorem}

\begin {proof} Let 
\bee
U_1 &:=& U_1 (\psi_1) = \left [ \frac {F}{\hhbar} \Pi_\perp \W \psi_1 + \frac {\eta}{\hhbar} \Pi_\perp |\psi_1^{2}| \psi_1 \right ] \\ 
U_2 &:=& U_2 (\psi_1 , \psi_\perp) = \left [ \frac {F}{\hhbar} \Pi_\perp \W \psi_\perp + \frac {\eta}{\hhbar} \Pi_\perp \left ( |\psi^2| \psi - 
|\psi_1^{2}| \psi_1 \right ) \right ]
\eee
then the previous equation (\ref {Eq15Ter}) becomes
\bee
\psi_\perp (\tau ) = f_1 (\tau ) + f_2 (\tau) 
\eee 
where 
\bee 
f_j (\tau ) = -i \int_0^\tau e^{-i (H_B-\Lambda_1 ) (\tau -\q )/\hhbar } U_j d\q  \, , \ j=1,2 \, , 
\eee
are such that

\begin {lemma}
The following estimates hold true:
\be
\left \| f_1 \right \|_{L^2} \le C \lambda + \tau C \hhbar^{-1} \lambda  \max \left [ \beta , \hhbar \lambda \right ]  \label {Eq13}
\ee
and
\be
\left \| f_2 \right \|_{L^2} \le C \lambda \int_0^\tau \| \psi_\perp (\q ) \|_{L^2} d\q  \, . \label {Eq14}
\ee
\end {lemma}

\begin {proof}
In order to prove the estimates (\ref {Eq13}) and (\ref {Eq14}) we remark that $e^{-i(H_B-\Lambda_1 )(\tau -s)/\hhbar }$ is an unitary operator; hence,  
\bee
\left \| e^{-i (H_B-\Lambda_1 ) (\tau -s)/\hhbar } U_j \right \|_{L^2} = \left \|  U_j \right \|_{L^2}\, , 
\ j =1,2\, . 
\eee
Now,
\bee
\| U_2 \|_{L^2} &\le & \frac {F}{\hhbar} \| \Pi_\perp W \psi_\perp \|_{L^2} + \frac {|\eta |}{\hhbar} \left \| \Pi_\perp \left ( 
|\psi |^2 \psi - |\psi_1 |^2 \psi_1 \right ) \right \|_{L^2} \\ 
&\le & C \hhbar^{-1} F  \| \psi_\perp \|_{L^2} + C \hhbar^{-1} |\eta| \left [ \| \psi_1 \psi_\perp^2 \|_{L^2} +  \| \psi_1^2 \psi_\perp \|_{L^2} +  
\| \psi_\perp^3 \|_{L^2} \right ] \\ 
&\le & C  \hhbar^{-1} F \| \psi_\perp \|_{L^2} + C \hhbar^{-1} |\eta| \left [ \| \psi_1 \|_{L^\infty} \| \psi_\perp \|_{L^\infty} +  
\| \psi_1 \|_{L^\infty}^2 +  \| \psi_\perp \|_{L^\infty}^2 \right ] \| \psi_\perp \|_{L^2} \\ 
&\le & C \left ( \hhbar^{-1}F + \hhbar^{-3/2} |\eta | \right )  \| \psi_\perp \|_{L^2} = C \lambda \| \psi_\perp \|_{L^2} 
\eee
from Theorem \ref {Teo1} and Corollary \ref {Cor1}; hence, (\ref {Eq14}) follows. \ In order to prove (\ref {Eq13}) we make use of an integration by parts:
\bee
f_1 (\tau ) &=& -i \int_0^\tau e^{-i (H_B-\Lambda_1 ) (\tau -\q )/\hhbar } U_1 d \q \\ 
&=& \left [ - \hhbar e^{-i (H_B-\Lambda_1 ) (\tau -\q )/\hhbar } [H_B-\Lambda_1 ]^{-1} U_1 \right ]_0^\tau +  
\hhbar \int_0^\tau e^{-i (H_B-\Lambda_1 ) (\tau -\q )/\hhbar }
[H_B-\Lambda_1 ]^{-1} \frac {d U_1}{d \q} d \q
\eee
From this fact and since $\| [H_B-\Lambda_1 ]^{-1} \Pi_\perp \| = \left [ \mbox {dist}(\Lambda_1, E_2^b) \right ]^{-1}
\sim \hhbar^{-1}$ then it follows that 
\bee
\| f_1 \|_{L^2} &\le & \max  \left [ \| U_1 (\tau ) \|_{L^2} , \| U_1 (0 ) \|_{L^2} \right ] + \tau \max_{s\in [0,\tau ] } 
\left \| \frac {dU_1}{d\tau } \right \|_{L^2} \\ 
&\le & C \lambda + \tau C \hhbar^{-1} \lambda  \max \left [ \beta , \hhbar \lambda, \| \Vector {r} \|_{\ell^2} \right ] \\ 
&\le & C \lambda + \tau C \hhbar^{-1} \lambda  \max \left [ \beta , \hhbar \lambda \right ] 
\eee
since 
\bee
\| U_1 \|_{L^2} &\le &  {F}{\hhbar^{-1}} \| \Pi_\perp W \psi_1 \|_{L^2} + {|\eta |}{\hhbar}^{-1} \| \Pi_\perp |\psi_1 |^2 \psi_1 \|_{L^2} \\ 
&\le & C  {F}{\hhbar^{-1}} \| \psi_1 \|_{L^2} + {|\eta |}{\hhbar}^{-1} \| \psi_1 \|_{L^\infty}^2 \| \psi_1 \|_{L^2} \\ 
&\le &C \lambda
\eee
and
\bee
\left \| \frac {dU_1}{d\tau } \right \|_{L^2} &\le & C \lambda \| \dot \psi_1 \|_{L^2} 
\le C \lambda \| \dot {\Vector {c}} \|_{\ell^2}  \\ 
& \le & C \hhbar^{-1} \lambda  \max \left [ \beta , \hhbar \lambda, \| \Vector {r} \|_{\ell^2} \right ]
\eee
by Lemma \ref {Lem7}, Theorem \ref {Teo2}, and from the draft estimate $\| \psi_\perp \|_{L^2} \le 1$.  \end {proof}

Hence, we have the following integral inequality
\bee
\| \psi_\perp \|_{L^2} \le  \left \{ C \lambda + \tau C \hhbar^{-1} \lambda  \max \left [ \beta ,  \hhbar \lambda \right ] 
\right \} + C \lambda \int_0^{\tau} \| \psi_\perp \|_{L^2} d\q 
\eee
and then the Gronwall's Lemma implies that
\bee
\| \psi_\perp \|_{L^2} \le  \left \{ C \lambda + \tau C \hhbar^{-1} \lambda  \max \left [ \beta , \hhbar 
\lambda \right ] \right \} e^{C \lambda  \tau}  \, , \ \forall \tau \in \R \, ;
\eee
Theorem \ref {Teo3} is so proved. 
\end {proof}

Now, we deal with the first differential equation of the system (\ref {Eq11})
\bee
\left \{
\begin {array}{l}
i\hhbar \dot c_n = G_n (\Vector {c}) + r_n  \, ,  \\ 
c_n (0) = \langle u_n , \psi_0 \rangle \, , \ 
\end {array}
\right.
\eee
where  
\bee G_n (\Vector {c}) = 
- \beta (c_{n+1} + c_{n-1}) + F \Chi_n c_n + \eta C_1 |c_n|^2 c_n \, . 
\eee
We compare it with the equation 
\be
\left \{
\begin {array}{l}
i\hhbar \dot g_n = G_n (\Vector {g})  \\ 
g_n (0) = c_n (0) 
\end {array}
\right. \label {Eq15Bis}
\ee
which represents the tight-binding approximation of (\ref {Eq3}), up to a phase factor $e^{-i\Lambda_1 \tau /\hhbar}$ depending on time. 

The Cauchy problem (\ref {Eq15Bis}) is globally well-posed, that is there exists a unique solution $\Vector {g} \in C (\R , \ell^2 (\Z ))$ that depends 
continuously on the the initial data (see, e.g. Theorem 1.3 \cite {P}). \ We must underline that we 
have the following a priori estimate $\| \Vector {c} \|_{\ell^2} \le 1$ and the conservation of the norm of $\Vector {g}$ 
\bee
\| \Vector {g} \|_{\ell^2} = \| \Vector {g}(0) \|_{\ell^2} = \| \Vector {c}(0) \|_{\ell^2} = 1 \, ; 
\eee
indeed, an immediate calculus gives that
\bee
\frac {d}{d\tau } i \hhbar \| \Vector {g} \|_{\ell^2}^2 = \beta \left [ \sum_{n\in \Z} \bar g_{n+1} g_n + \sum_{n\in \Z} \bar g_{n-1} g_n - 
\sum_{n\in \Z} g_{n+1} \bar g_n - \sum_{n\in \Z} g_{n-1} \bar g_n \right ] =0
\eee
because $\beta$, $F$, $\eta $, $C_1$ and $\xi_n$ are real-valued.

Then, it follows that the vector $\Vector {c}-\Vector {g}$ satisfies to the following integral equation
\bee
i \hhbar \left [c_n (\tau ) - g_n (\tau ) \right ] = \int_0^\tau \left [ G_n (\Vector {c}) - G_n (\Vector {g}) \right ] d\q  + \int_0^\tau r_n (\q ) d\q  \, , 
\eee
from which 
\bee
\| \Vector {c} - \Vector {g} \|_{\ell^2} \le \frac {1}{\hhbar} \int_0^\tau \| \Vector {G}(\Vector {c}) - \Vector {G}(\Vector {g}) \|_{\ell^2} d\q + \frac {1}{\hhbar} 
\int_0^\tau \| \Vector {r} \|_{\ell^2} d\q \, .
\eee

\begin {lemma}
$ \Vector {G}$ is a Lipschitz function such that 
\be
\| \Vector {G}(\Vector {c}) - \Vector {G}(\Vector {g}) \|_{\ell^2} \le C \max [\beta , \hhbar \lambda ] \|\Vector {c} - \Vector {g} \|_{\ell^2}\, . \label {F3}
\ee
\end {lemma}

\begin {proof}
Indeed
\bee
&& \| \Vector {G}(\Vector {c}) - \Vector {G}(\Vector {g}) \|_{\ell^2}^2 = \sum_{n\in \Z} \left | G_n (\Vector {c} ) - G_n (\Vector {g} ) \right |^2 = \\ 
&&\ =  \sum_{n\in \Z} \left | - \beta \left [ \left ( c_{n+1}-g_{n+1} \right ) + 
\left ( c_{n-1}-g_{n-1} \right ) \right ] + F \Chi_n (c_n - g_n ) + \right. \\ 
&& \left. \ \ + \eta C_1 \left [ |c_n|^2 c_n - |g_n|^2 g_n \right ] \right |^2 \le \\ 
&& \ \le C \left \{ \beta^2 \sum_{n\in \Z} \left | c_{n+1}-g_{n+1} \right |^2 +  \beta^2 \sum_{n\in \Z} \left | c_{n-1}-g_{n-1} \right |^2 + \sum_{n\in \Z} F^2 \Chi_n^2  
|c_n - g_n |^2 + \right. \\ 
&& \left. \ \ \ + \sum_{n\in \Z} \eta^2 C_1^2 \left | |c_n|^2 c_n - |g_n|^2 g_n \right | \right \} \\ 
&& \ \le C \left \{2 \beta^2 \left \| \Vector {c} - \Vector {g} \right \|_{\ell^2}^2 + \left [ \max_{n\in \Z} \Chi_n^2 \right ] F^2  \left \| \Vector {c} - \Vector {g} 
\right \|_{\ell^2}^2 + \right. \\ 
&& \left. \ \ \eta^2 C_1^2 \sum_{n\in \Z} \left | |c_n|^2 (c_n -g_n) + \left ( |c_n|^2 - |g_n|^2 \right ) |g_n| \right | \right \} \\ 
 && \ \le C \left [ \beta^2 + F^2 + \eta^2 C_1^2 \right ] \left \| \Vector {c} - \Vector {g} \right \|_{\ell^2}^2 
 \le C \max [\beta^2 , \hhbar^2 \lambda^2 ] \left \| \Vector {c} - \Vector {g} 
 \right \|_{\ell^2}^2
 \eee
since $ |c_n|^2-|g_n|^2 = c_n \left [ \bar c_n - \bar g_n \right ] + \bar g_n \left [ c_n - g_n \right ] $ and (\ref {Eq10}). \end {proof}

By Theorems \ref {Teo2} and \ref {Teo3}, it turns out that the vector $\Vector {r}$ is norm bounded by 
\be
 \| \Vector {r} \|_{\ell^2} \le a + b e^{C\lambda \tau } + c \tau e^{C\lambda \tau } \label {F4}
 \ee
for some positive constant $C$ independent of $\hhbar$ and where
\be
a &=& \omega + C \lambda e^{- (S_0 - \rho )/\hhbar} \, , \label {Eqa} \\
b &=& C \hhbar \lambda^2 \, , \label {Eqb} \\
c &=& C\lambda^2 \max \left [ \beta , \hhbar \lambda \right ] \, . \label {Eqc}
\ee
Then, we get the integral inequality
\be
\| \Vector {c} - \Vector {g} \|_{\ell^2} \le \alpha (\tau ) + \int_0^\tau \delta (\q ) \| \Vector {c} - \Vector {g} \|_{\ell^2}  d\q \label {F5}
\ee
where
\bee
\alpha (\tau ) \le  \frac {1}{\hhbar} \int_0^\tau \| \Vector {r} \|_{\ell^2} d\q \le C \left [ \frac {a\tau}{\hhbar} + \frac {C b\lambda + c}{C^2 \lambda^2 \hhbar} 
\left ( e^{C \lambda \tau}-1 \right ) + \frac {c\tau }{C \lambda \hhbar} e^{C\lambda \tau} \right ]
\eee
and
\bee
\delta (\tau )\le C \max [\hhbar^{-1} \beta , \lambda ]
\eee
By the Gronwall's Lemma we finally get the estimate
\bee
\| \Vector {c} - \Vector {g} \|_{\ell^2} &\le & \alpha (\tau ) e^{\int_0^\tau \delta (\q ) d\q  } = 
\alpha (\tau ) e^{C \max [\hhbar^{-1} \beta , \lambda ] \tau}
\eee

Therefore, we have proved that 
\begin {lemma} \label {Lem10}
Let $a$, $b$ and $c$ dafined by (\ref {Eqa}-\ref {Eqc}), then  Equazioni Differenziali della Fisica Matematica
\be
\| \Vector {c} - \Vector {g} \|_{\ell^2} \le C \left \{ \frac {a\tau}{\hhbar} + \frac {C b\lambda + c}{C^2 \lambda^2 \hhbar} \left ( e^{C \lambda \tau}-1 \right ) + 
\frac {c\tau }{C \lambda \hhbar} e^{C\lambda \tau} \right \} e^{C \max [\hhbar^{-1} \beta , \lambda ] \tau}
\label {F6}
\ee
for some positive constant $C$ independent of $\hhbar$.
\end {lemma}

In conclusion

\begin {theorem} \label {Teo4}
Let $\Vector {g} \in C (\R , \ell^2 (\Z ))$ be the solution to the discrete nonlinear Schr\"odinger equation (\ref {Eq15Bis}); let $\psi (\tau , x)\in C (\R , H^1 (\R ))$ 
be the solution to the nonlinear Schr\"odinger equation (\ref {Eq3}) with initial condition $\psi_0 (x)= \sum g_n (0) \w_n (x)$; 
let $a$, $b$ and $c$ defined by Lemma \ref {Lem10}; let $\lambda$ be defined by 
Lemma \ref {Lem7}. \ Then, for some positive constant $C$ independent of $\hhbar$ it follows that 
\be
&& \left \| \psi (\tau , \cdot ) - \sum_{n\in \Z} g_n (\tau) e^{i\Lambda_1 \tau /\hhbar } u_n (\cdot ) \right \|_{L^2} \le \\ 
&& \ \ \le \left \{ C \lambda + \tau C \hhbar^{-1} \lambda  \max \left [ \beta , \hhbar \lambda \right ] \right \} 
e^{C \lambda  \tau} + \nonumber \\ 
&& \ + C \left \{ \frac {a\tau}{\hhbar} + \frac {C b\lambda + c}{C^2 \lambda^2 \hhbar} \left ( e^{C \lambda \tau}-1 \right ) + 
\frac {c\tau }{C \lambda \hhbar} e^{C\lambda \tau} \right \} e^{C \max [\hhbar^{-1} \beta , \lambda ] \tau} \, , \ \forall \tau \in \R^+ \, .  \label {F9}
\ee
\end {theorem}

\begin {proof}
Indeed, recalling that we made use of the gauge choice $\psi \to e^{i \Lambda_1 \tau /\hhbar} \psi$, we have that 
\bee
\left \| \psi  - \sum_{n\in \Z} g_n  e^{i\Lambda_1 \tau /\hhbar } \w_n  \right \|_{L^2}^2 &=& 
\left \| e^{-i\Lambda_1 \tau /\hhbar } \psi  - \sum_{n\in \Z} g_n   \w_n  \right \|_{L^2}^2 \\ 
&=& \| \psi_\perp \|_{L^2}^2 + \left \| \sum_{n\in \Z} \left ( c_n - g_n  
 \right ) \w_n \right \|_{L^2}^2
\eee
where 
\bee
\left \| \sum_{n\in \Z} \left ( c_n - g_n  
 \right ) \w_n \right \|_{L^2}^2 = \sum_{n\in Z} |c_n-g_n|^2 \| \w_n \|_{L^2}^2 = \| \Vector {c} - \Vector {g} \|_{\ell^2}^2 
 \eee
 because $\{ \w_n \}$ is an orthonormal set of vectors.
\end {proof}

\subsection {Proof of Corollary \ref {Cor1Bis}}\label {Sec4.1}

Here we assume, according with {Hypothesis 3a)}, that the real-valued parameters $\alpha_1 $ and $\alpha_2$ are fixed; in such a case we have that 
\be
F \sim \eta \sim \hhbar^2 \, . \label {Eq4}
\ee
Therefore:
\bee
\lambda \sim \hhbar^{1/2} \, , \ |a|, \beta \le Ce^{-(S_0-\rho )/\hhbar} \, , \  b \sim \hhbar^{2} \, , \ c \sim \hhbar^{5/2} \, . 
\eee
Then the estimate (\ref {F9}) makes sense for times of order $\tau \in [0, C \hhbar^{-\gamma} ]$ for some fixed $\gamma \le \frac 12$. \ In such an interval 
we have that 
\bee
 C \lambda + \tau C \hhbar^{-1} \lambda  \max \left [ \beta ,  \hhbar \lambda \right ] \sim \hhbar^{1/2} + 
 \hhbar^{1 - \gamma}  \\ 
 \frac {a\tau}{\hhbar} + \frac {C b\lambda + c}{C^2 \lambda^2 \hhbar} \left ( e^{C \lambda \tau}-1 \right ) + 
\frac {c\tau }{C \lambda \hhbar} e^{C\lambda \tau} \sim \hhbar^{1/2} + \hhbar^{1-\gamma}\, . 
\eee

In particular, for $\gamma = \frac 12$ Corollary \ref {Cor1Bis} follows.

\subsection {Proof of Corollary \ref {Cor2}}\label {Sec4.2}

Here we assume, according with {Hypothesis 3b)}, that the real-valued parameters $\alpha_1 $ and $\alpha_2$ are not fixed, but both go to zero when $\epsilon $ 
goes to zero; in particular we have that
\bee
F \sim \hhbar^{-1/2} \eta \sim \beta \, . 
\eee
In such a case we have that
\bee
\lambda \sim \hhbar^{-1} \beta \, , \ a \sim \omega \, , \ b \sim  \hhbar^{-1} \beta^2 \, , \ c \sim  \hhbar^{-2} \beta^3 \, . 
\eee
The estimate (\ref {F9}) makes sense for times of order $\tau \in [0, \beta^{-1} \hhbar ]$. \ In such an interval we have that 
\bee
\| \psi_\perp \|_{L^2} \le C e^{-(S_0 - \rho )/\hhbar} 
\eee
and 
\bee
\| \Vector {c} - \Vector {g} \|_{\ell^2} \le C e^{-\zeta /\hhbar}  
\eee
for some $\zeta >0$. \ Hence, Corollary \ref {Cor2} is proved.

\end {document}